\documentclass[11pt]{article}%
\usepackage{amsmath}
\usepackage{amsfonts}
\usepackage{amssymb}
\usepackage{graphicx}
\usepackage{fullpage}%
\providecommand{\U}[1]{\protect\rule{.1in}{.1in}}
\newtheorem{theorem}{Theorem}

\newtheorem{corollary}[theorem]{Corollary}

\newtheorem{lemma}[theorem]{Lemma}

\newenvironment{proof}[1][Proof]{\noindent\textbf{#1.} }{\ \rule{0.5em}{0.5em}}
\begin{document}

\title{A Linear-Optical Proof that the Permanent is $\mathsf{\#P}$-Hard}
\author{Scott Aaronson\thanks{MIT. \ Email: aaronson@csail.mit.edu. \ \ This material
is based upon work supported by the National Science Foundation under Grant
No. 0844626. \ Also supported by a DARPA YFA grant and a Sloan Fellowship.}}
\date{\textit{For Les Valiant, on the occasion of his Turing Award}}
\maketitle

\begin{abstract}
One of the crown jewels of complexity theory is Valiant's 1979 theorem that
computing the permanent of an $n\times n$\ matrix\ is $\mathsf{\#P}$-hard.
\ Here we show that, by using the model of \textit{linear-optical quantum
computing}---and in particular, a universality theorem due to Knill, Laflamme,
and Milburn---one can give a different and arguably more intuitive proof of
this theorem.

\end{abstract}

\section{Introduction\label{INTRO}}

Given an $n\times n$\ matrix $A=\left(  a_{i,j}\right)  $,\ the
\textit{permanent} of $A$ is defined as%
\[
\operatorname*{Per}\left(  A\right)  =\sum_{\sigma\in S_{n}}\prod_{i=1}%
^{n}a_{i,\sigma\left(  i\right)  }.
\]
A seminal result of Valiant \cite{valiant}\ says that computing
$\operatorname*{Per}\left(  A\right)  $ is $\mathsf{\#P}$-hard, if $A$ is a
matrix over (say) the integers, the nonnegative integers, or the set $\left\{
0,1\right\}  $.\footnote{See Hrubes, Wigderson, and Yehudayoff \cite{hwy} for
a recent, \textquotedblleft modular\textquotedblright\ presentation of
Valiant's proof (which also generalizes the proof to the noncommutative and
nonassociative case).} \ Here $\mathsf{\#P}$\ means (informally) the class of
\textit{counting problems}---problems that involve summing exponentially-many
nonnegative integers---and $\mathsf{\#P}$-hard\ means \textquotedblleft at
least as hard as any $\mathsf{\#P}$ problem.\textquotedblright\footnote{See
the Complexity Zoo (www.complexityzoo.com) for the definitions of
$\mathsf{\#P}$\ and other complexity classes used in this paper.}$^{,}%
$\footnote{If $A$ is a nonnegative integer matrix, then $\operatorname*{Per}%
\left(  A\right)  $\ is \textit{itself} a $\mathsf{\#P}$\ function, which
implies that it is $\mathsf{\#P}$\textit{-complete}\ (the term for functions
that are both $\mathsf{\#P}$-hard and in $\mathsf{\#P}$). \ If $A$ can have
negative or fractional entries, then strictly speaking $\operatorname*{Per}%
\left(  A\right)  $\ is no longer $\mathsf{\#P}$-complete, but it is still
$\mathsf{\#P}$-hard\ and computable in the class $\mathsf{FP}^{\mathsf{\#P}}%
$.}

More concretely, Valiant gave a polynomial-time algorithm that takes as input
an instance $\varphi\left(  x_{1},\ldots,x_{n}\right)  $\ of the Boolean
satisfiability problem, and that outputs a matrix $A_{\varphi}$\ such that
$\operatorname*{Per}\left(  A_{\varphi}\right)  $\ encodes the number of
satisfying assignments of $\varphi$. \ This means that computing the permanent
is at least as hard as counting satisfying assignments.

Unfortunately, the standard proof that the permanent is $\mathsf{\#P}$-hard is
notoriously opaque; it relies on a set of gadgets that seem to exist for
\textquotedblleft accidental\textquotedblright\ reasons. \ Could there be an
alternative proof that gave more, or at least different, insight? \ In this
paper, we try to answer that question by giving a new, quantum-computing-based
proof that the permanent is $\mathsf{\#P}$-hard. \ In particular, we will
derive the permanent's $\mathsf{\#P}$-hardness as a consequence of the
following three facts:

\begin{enumerate}
\item[(1)] \textit{Postselected linear optics} is capable of universal quantum
computation, as shown in a celebrated 2001 paper of Knill, Laflamme, and
Milburn\ \cite{klm} (henceforth referred to as KLM).\footnote{KLM actually
prove the stronger (and more practically-relevant) result that linear optics
with \textit{adaptive measurements} is capable of universal quantum
computation. \ For our purposes, however, we only need the weaker fact that
\textit{postselected} measurements suffice for universal QC, which KLM prove
as a lemma along the way to their main result.}

\item[(2)] Quantum computations can encode $\mathsf{\#P}$-hard\ quantities in
their amplitudes.

\item[(3)] Amplitudes in $n$-photon linear-optics circuits can be expressed as
the permanents of $n\times n$\ matrices.
\end{enumerate}

Even though our proof is based on quantum computing, we stress that we have
made it \textit{entirely self-contained}: all of the results we need
(including the KLM Theorem \cite{klm}, and even the construction of the
Toffoli gate from $1$-qubit and $\operatorname*{CSIGN}$\ gates) are proved in
this paper for completeness. \ We assume some familiarity with quantum
computing notation (e.g., kets and quantum circuit diagrams), but not with
linear optics.

\subsection{Motivation\label{MOTIV}}

If one counts the complexity of all of the individual pieces we
use---especially the universality results for quantum gates---then our
reduction from $\mathsf{\#P}$\ to the permanent ends up being at least as
complicated as Valiant's, and probably more so. \ In our view, however, this
is similar to how writing a program in C++ tends to produce a longer, more
complicated executable file than writing the same program in assembly
language. \ Normally, one also cares about the length and readability of the
\textit{source code}! \ Our purpose in this paper is to illustrate how quantum
computing provides a powerful \textquotedblleft high-level programming
language\textquotedblright\ in which one can, among other things, easily
rederive the most celebrated result in the theory of $\mathsf{\#P}$-hardness.

But why does the world need a new proof that the permanent is $\mathsf{\#P}%
$-hard---especially a proof invoking what some might consider to be exotic
concepts? \ Let us offer several answers:

\begin{itemize}
\item Any theorem as basic as the $\mathsf{\#P}$-hardness of the permanent
deserves several independent proofs. \ And our proof really is
\textquotedblleft independent\textquotedblright\ of the standard one: rather
than composing variable and clause gadgets,\footnote{Indeed, our proof does
not even go through the Cook-Levin Theorem: it reduces a $\mathsf{\#P}%
$\ computation directly to the permanent, without first reducing
$\mathsf{\#P}$\ to $\#3SAT$.} we \textit{multiply matrices} corresponding to
quantum gates, and use ideas from linear optics to keep track of how such
multiplications affect the permanent. \ One way to see the difference is that
our proof never uses the notion of a cycle cover.

\item While our proof, like the standard one, requires \textquotedblleft
gadgets\textquotedblright\ (one to simulate a Toffoli gate using
$\operatorname*{CSIGN}$\ gates, another to simulate a $\operatorname*{CSIGN}$
gate using postselected linear optics), the connection to quantum computing
gives those gadgets a natural \textit{semantics}. \ In other words, the
gadgets were introduced for \textquotedblleft practical\textquotedblright%
\ reasons having nothing to do with proving the permanent $\mathsf{\#P}$-hard,
and can be motivated independently of that goal. \ \textit{If} one already
knows the quantum universality gadgets, \textit{then} we offer what seems like
a major advance in complexity-theoretic pedagogy: a proof that the permanent
is $\mathsf{\#P}$-hard that can be reproduced on-the-spot from memory!

\item As Kuperberg \cite{kuperberg:jones}\ pointed out, by their nature, any
$\mathsf{\#P}$-hardness\ proofs (including ours) that are based on
\textquotedblleft quantum postselection\textquotedblright\ almost
\textit{immediately} yield hardness of approximation results as well.

\item We expect that the quantum postselection\ approach used here could lead
to $\mathsf{\#P}$-hardness proofs for many other problems---including problems
not already known to be $\mathsf{\#P}$-hard by other means. \ In this
direction, one natural place to look would be special cases of the permanent.
\end{itemize}

\subsection{Related Work\label{RELATED}}

By now, there are many examples where quantum computing has been used to give
new or simpler proofs of classical complexity theorems; see Drucker and de
Wolf \cite{ddw}\ for an excellent survey. \ Within the area of counting
complexity, Aaronson \cite{aar:pp} showed that the class $\mathsf{PP}$ is
equal to $\mathsf{PostBQP}$\ (quantum polynomial-time with postselection), and
then used that theorem to give a simpler proof of the landmark result of
Beigel, Reingold, and Spielman \cite{brs}\ that $\mathsf{PP}$\ is closed under
intersection. \ Later, also using the $\mathsf{PostBQP}=\mathsf{PP}$ theorem,
Kuperberg \cite{kuperberg:jones} gave a \textquotedblleft quantum
proof\textquotedblright\ of the result of Jaeger, Vertigan, and Welsh
\cite{jvw}\ that computing the Jones polynomial is $\mathsf{\#P}$-hard, and
even showed that a certain approximate version is $\mathsf{\#P}$-hard (which
had not been shown previously). \ Kuperberg's argument for the Jones
polynomial is conceptually similar to our argument for the permanent.

There is also precedent for using linear optics as a tool to prove theorems
about the permanent. \ Scheel \cite{scheel} observed that the unitarity of
linear-optical quantum computing implies the interesting fact that $\left\vert
\operatorname*{Per}\left(  U\right)  \right\vert \leq1$\ for all unitary
matrices $U$.

Rudolph \cite{rudolph} showed how to encode quantum amplitudes directly as
matrix permanents, and in the process, gave a \textquotedblleft
quantum-computing proof\textquotedblright\ that the permanent is
$\mathsf{\#P}$-hard. \ However, a crucial difference is that Rudolph starts
with Valiant's proof based on cycle covers, then recasts it in quantum terms
(with the goal of making Valiant's proof more accessible to a physics
audience). \ By contrast, our proof is independent of Valiant's; the tools we
use were invented for separate reasons in the quantum computing literature.

There has been a great deal of work on linear-optical quantum computing,
beyond the seminal KLM Theorem \cite{klm} on which this paper relies.
\ Recently, Aaronson and Arkhipov \cite{aark} studied the complexity of
sampling from a linear-optical computer's output distribution, assuming no
adaptive measurements are available. \ By using the $\mathsf{\#P}$-hardness of
the permanent as an \textquotedblleft input axiom,\textquotedblright\ they
showed that this sampling problem is classically intractable unless
$\mathsf{P}^{\mathsf{\#P}}=\mathsf{BPP}^{\mathsf{NP}}$. \ More relevant to
this paper is an alternative \textit{proof} that Aaronson and Arkhipov gave
for their result.\ \ Inspired by work of Bremner, Jozsa, and Shepherd
\cite{bjs}, the alternative proof combines Aaronson's $\mathsf{PostBQP}%
=\mathsf{PP}$\ theorem \cite{aar:pp}\ with the fact that postselected linear
optics is universal for $\mathsf{PostBQP}$, and thereby avoids any direct
appeal to the $\mathsf{\#P}$-hardness of the permanent. \ In retrospect, that
proof was already much of the way toward a linear-optical proof that the
permanent is $\mathsf{\#P}$-hard; this paper simply makes the connection explicit.

\section{Background\label{PRELIM}}

Not by accident, this section constitutes the bulk of the paper. \ First, in
Section \ref{CIRCUITS}, we fix some facts and notation about standard
(qubit-based) quantum computing. \ Then, in Section \ref{LOQC}, we give a
short overview of those aspects of linear-optical quantum computing that are
relevant for us, and (for completeness) prove the KLM Theorem in the specific
form we will need.

\subsection{Quantum Circuits\label{CIRCUITS}}

Abusing notation, we will often identify a quantum circuit $Q$ with the
unitary transformation that it induces: for example, $\left\langle
0\cdots0\right\vert Q\left\vert 0\cdots0\right\rangle $\ represents the
amplitude with which $Q$ maps its initial state to itself. \ We use
$\left\vert Q\right\vert $\ to denote the number of gates in $Q$.

The first ingredient we need for our proof is a convenient set of quantum
gates (in the standard qubit model). \ Thus, let $\mathcal{G}$\ be the set of
gates consisting of (1) all $1$-qubit gates, and (2) the $2$-qubit
\textit{controlled-sign} gate%
\[
\operatorname*{CSIGN}=\left(
\begin{array}
[c]{cccc}%
1 & 0 & 0 & 0\\
0 & 1 & 0 & 0\\
0 & 0 & 1 & 0\\
0 & 0 & 0 & -1
\end{array}
\right)  ,
\]
which flips the amplitude if and only if both qubits are\ $\left\vert
1\right\rangle $.\footnote{A more common $2$-qubit gate than
$\operatorname*{CSIGN}$\ is the \textit{controlled-NOT} ($\operatorname*{CNOT}%
$) gate, which maps each basis state $\left\vert x,y\right\rangle $\ to
$\left\vert x,y\oplus x\right\rangle $. \ However, $\operatorname*{CSIGN}$\ is
more convenient for linear-optics purposes, and is equivalent to
$\operatorname*{CNOT}$\ by conjugating the second qubit with a Hadamard gate.}
\ Then Barenco et al.\ \cite{barenco} showed that $\mathcal{G}$\ is a
\textit{universal} set of quantum gates, in the sense that $\mathcal{G}$
generates \textit{any} unitary transformation on any number of qubits (without
error). \ For our purposes, however, the following weaker result suffices.

\begin{lemma}
\label{csignlem}$\mathcal{G}$ generates the Toffoli gate, the $3$-qubit gate
that maps each basis state $\left\vert x,y,z\right\rangle $\ to $\left\vert
x,y,z\oplus xy\right\rangle $.
\end{lemma}

\begin{proof}
The circuit can be found in Nielsen and Chuang \cite{nc}\ for example, but we
reproduce it in Figure \ref{toffolifig} for completeness. \ In the diagram,%
\[
H=\frac{1}{\sqrt{2}}\left(
\begin{array}
[c]{cc}%
1 & 1\\
1 & -1
\end{array}
\right)
\]
is the Hadamard gate,%
\[
B=\frac{1}{2}\left(
\begin{array}
[c]{cc}%
\sqrt{2+\sqrt{2}} & i\sqrt{2-\sqrt{2}}\\
i\sqrt{2-\sqrt{2}} & \sqrt{2+\sqrt{2}}%
\end{array}
\right)
\]
is another $1$-qubit gate, and the six vertical bars represent
$\operatorname*{CSIGN}$\ gates.%
\begin{figure}[ptb]%
\centering
\includegraphics[
trim=1.237065in 6.190364in 1.238435in 0.000000in,
height=0.8354in,
width=4.5766in
]%
{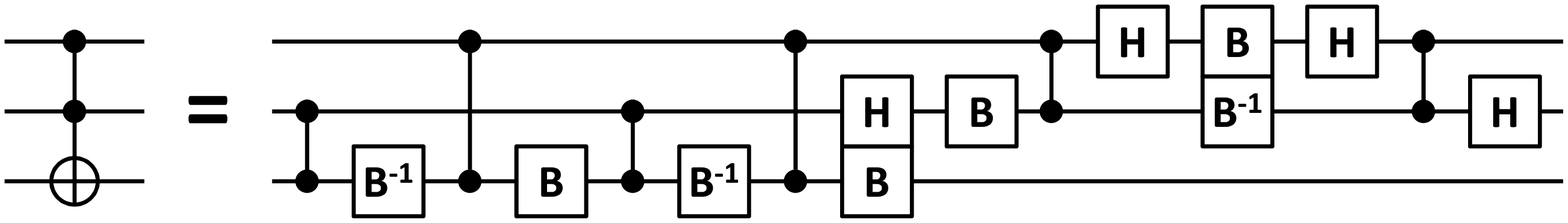}%
\caption{Simulating a Toffoli gate using $\operatorname*{CSIGN}$ and $1$-qubit
gates.}%
\label{toffolifig}%
\end{figure}

\end{proof}

\subsection{Linear-Optical Quantum Computing\label{LOQC}}

We now give a brief overview of \textit{linear-optical quantum computing}
(LOQC), an alternative quantum computing model based on identical photons
rather than qubits. \ For a detailed introduction to LOQC from a computer
science perspective, see Aaronson and Arkhipov \cite{aark}.

In LOQC, each basis state of our quantum computer has the form $\left\vert
S\right\rangle =\left\vert s_{1},\ldots,s_{m}\right\rangle $, where
$s_{1},\ldots,s_{m}$\ are nonnegative integers summing to $n$. \ Here $s_{i}%
$\ represents the number of photons in the $i^{th}$\ location or
\textquotedblleft mode,\textquotedblright\ and\ the fact that $s_{1}%
+\cdots+s_{m}=n$\ means that photons are never created or destroyed. \ One
should think of $m$\ and $n$\ as both polynomially-bounded. \ For this paper,
it will be convenient to assume that $m$ is even, that $n=m/2$, and that the
initial state has the form $\left\vert I\right\rangle =\left\vert
0,1,0,1,\ldots,0,1\right\rangle $: that is, one photon in each even-numbered
mode, and no photons in the odd-numbered modes.

Let $\Phi_{m,n}$\ be the set of nonnegative integer tuples $S=\left(
s_{1},\ldots,s_{m}\right)  $\ such that $s_{1}+\cdots+s_{m}=n$, and let
$\mathcal{H}_{m,n}$\ be the Hilbert space spanned by basis states $\left\vert
S\right\rangle $\ with $S\in\Phi_{m,n}$. \ Then a general state in LOQC is
just a unit vector in $\mathcal{H}_{m,n}$:%
\[
\left\vert \psi\right\rangle =\sum_{S\in\Phi_{m,n}}\alpha_{S}\left\vert
S\right\rangle
\]
with $\sum_{S\in\Phi_{m,n}}\left\vert \alpha_{S}\right\vert ^{2}=1$.

To transform $\left\vert \psi\right\rangle $, one can select any $m\times
m$\ unitary transformation $U=\left(  u_{ij}\right)  $. \ This $U$ then
induces a larger unitary transformation $\varphi\left(  U\right)  $\ on the
Hilbert space $\mathcal{H}_{m,n}$ of $n$-photon states. \ There are several
ways to define $\varphi\left(  U\right)  $, but perhaps the simplest is the
following formula:%
\begin{equation}
\left\langle S|\varphi\left(  U\right)  |T\right\rangle =\frac
{\operatorname*{Per}\left(  U_{S,T}\right)  }{\sqrt{s_{1}!\cdots s_{m}%
!t_{1}!\cdots t_{m}!}} \tag{*}\label{perform}%
\end{equation}
for all tuples $S=\left(  s_{1},\ldots,s_{m}\right)  $\ and $T=\left(
t_{1},\ldots,t_{m}\right)  $ in $\Phi_{m,n}$. \ Here $U_{S,T}$\ is the
$n\times n$\ matrix obtained from $U$ by taking $s_{i}$\ copies of the
$i^{th}$\ row of $U$ and $t_{j}$\ copies of the $j^{th}$\ column,\ for all
$i,j\in\left[  m\right]  $. \ To illustrate, if%
\[
U=\left(
\begin{array}
[c]{cc}%
1 & 0\\
0 & -1
\end{array}
\right)
\]
and $\left\vert S\right\rangle =\left\vert T\right\rangle =\left\vert
2,1\right\rangle $, then%
\[
U_{S,T}=\left(
\begin{array}
[c]{ccc}%
1 & 1 & 0\\
1 & 1 & 0\\
0 & 0 & -1
\end{array}
\right)  .
\]

Intuitively, the reason the permanent arises in formula (*) is that there are
$n!$\ ways of mapping the $n$\ photons in basis state $\left\vert
S\right\rangle $\ onto the $n$ photons in basis state $\left\vert
T\right\rangle $. \ Since the photons are \textit{identical bosons}, quantum
mechanics says that each of those $n!$\ ways contributes a term to the total
$\left\langle S|\varphi\left(  U\right)  |T\right\rangle $, with the
contribution given by the product of the transition amplitudes $u_{ij}$\ for
each of the $n$ photons individually.

It turns out that $\varphi\left(  U\right)  $\ is always unitary and that
$\varphi$\ is a homomorphism. \ Both facts seem surprising viewed purely as
algebraic consequences of formula (*), but of course they have natural
physical interpretations: $\varphi\left(  U\right)  $\ is unitary because it
represents an actual physical transformation that can be applied, and
$\varphi$\ is a homomorphism because generalizing from one photon to $n$
photons must commute with composing beamsplitters. \ In this paper, we will
not need that $\varphi\left(  U\right)  $\ is unitary; see Aaronson and
Arkhipov \cite{aark}\ for a proof of that fact. \ Below we prove that
$\varphi$\ is a homomorphism.

\begin{lemma}
\label{homlem}$\varphi$ is a homomorphism.
\end{lemma}

\begin{proof}
We want to show that for all tuples $S,T\in\Phi_{m,n}$\ and all $m\times
m$\ unitaries $U,V$,%
\[
\left\langle S|\varphi\left(  VU\right)  |T\right\rangle =\left\langle
S|\varphi\left(  V\right)  \varphi\left(  U\right)  |T\right\rangle
=\sum_{R\in\Phi_{m,n}}\left\langle S|\varphi\left(  V\right)  |R\right\rangle
\left\langle R|\varphi\left(  U\right)  |T\right\rangle .
\]
By equation (*), the above is equivalent (after multiplying both sides by
$\sqrt{s_{1}!\cdots s_{m}!t_{1}!\cdots t_{m}!}$) to the identity%
\begin{equation}
\operatorname*{Per}\left(  \left(  VU\right)  _{S,T}\right)  =\sum_{R\in
\Phi_{m,n}}\frac{\operatorname*{Per}\left(  V_{S,R}\right)
\operatorname*{Per}\left(  U_{R,T}\right)  }{r_{1}!\cdots r_{m}!}. \tag{**}%
\end{equation}
We will prove identity (**) in the special case $n=m$\ and $S=T=I=\left(
1,1,\ldots,1\right)  $, since the general case is analogous. \ We have%
\begin{align*}
\operatorname*{Per}\left(  VU\right)   &  =\sum_{\sigma\in S_{n}}\prod
_{i=1}^{n}\left(  VU\right)  _{i,\sigma\left(  i\right)  }\\
&  =\sum_{R\in\Phi_{n,n}}\left(  \frac{1}{r_{1}!\cdots r_{n}!}\sum_{\tau
,\xi\in S_{n}}\prod_{i=1}^{n}\left(  V_{I,R}\right)  _{i,\xi\left(  i\right)
}\left(  U_{R,I}\right)  _{i,\tau\left(  i\right)  }\right) \\
&  =\sum_{R\in\Phi_{n,n}}\frac{\operatorname*{Per}\left(  V_{I,R}\right)
\operatorname*{Per}\left(  U_{R,I}\right)  }{r_{1}!\cdots r_{n}!}.
\end{align*}
In the second line above, we decomposed the sum by thinking about each
permutation $\sigma\in S_{n}$\ as a product of \textit{two} permutations: one,
$\tau$, that maps $n$ particles in the initial configuration $\left\vert
I\right\rangle $\ to $n$ particles in the intermediate configuration
$\left\vert R\right\rangle $ when $U$ is applied, and another, $\xi$, that
maps $n$ particles in the intermediate configuration $\left\vert
R\right\rangle $\ to $n$ particles in the final configuration $\left\vert
I\right\rangle $ when $V$\ is applied. \ This yields the same result, as long
as we remember to sum over all possible intermediate configurations $R\in
\Phi_{n,n}$, and also to divide each summand by $r_{1}!\cdots r_{n}!$, which
is the size of $R$'s automorphism group (i.e., the\ number of ways to permute
the $n$\ particles within $\left\vert R\right\rangle $\ that leave $\left\vert
R\right\rangle $\ unchanged). \ 
\end{proof}

In the standard qubit model, every unitary transformation can be decomposed as
a product of \textit{gates}, each of which acts nontrivially on only $1$ or
$2$ qubits. \ Similarly, in LOQC, every unitary transformation can be
decomposed as a product of \textit{linear-optics gates}, each of which acts
nontrivially on only $1$ or $2$ modes. \ Then a \textit{linear-optics circuit}
is simply a list of linear-optics gates applied to specified modes (or pairs
of modes) starting from the initial state $\left\vert I\right\rangle
=\left\vert 0,1,\ldots,0,1\right\rangle $.\footnote{A crucial difference
between standard quantum circuits and linear-optics circuits is that, whereas
a standard quantum gate is the \textit{tensor product} of a small (say
$4\times4$) unitary matrix with an exponentially-large (say $2^{n-2}%
\times2^{n-2}$) identity matrix, a linear-optics gate is the \textit{direct
sum} of a small (say $2\times2$) unitary matrix with a \textit{polynomially}%
-large (say $\left(  m-2\right)  \times\left(  m-2\right)  $) identity matrix.
It is only the homomorphism $U\rightarrow\varphi\left(  U\right)  $\ that
produces exponentially-large matrices. \ One consequence, pointed out by Reck
et al.\ \cite{rzbb}, is that, whereas most $n$-qubit unitary transformations
require $\Omega\left(  2^{2n}\right)  $ gates to implement (as follows from an
easy dimension argument), every $m$-mode unitary transformation $U$\ can be
implemented using only $O\left(  m^{2}\right)  $\ linear-optics gates.}

The last notion we need is that of \textit{postselected LOQC}. \ In our
context, postselection simply means measuring the number of photons in a given
mode $i$, and conditioning on a particular result (for example, $0$ photons,
or $1$ photon). \ After we postselect on the number of photons in some mode,
we will never use that mode for further computation.\footnote{In physics
language, all photon-number measurements are assumed to be \textquotedblleft
demolition\textquotedblright\ measurements.} \ For this reason, without loss
of generality, we can defer all postselected measurements until the end of the computation.

Our $\mathsf{\#P}$-hardness\ proof will fall out as a corollary of the
following universality theorem, which is implicit in the work of KLM
\cite{klm}. \ Indeed, we \textit{could} just appeal to the KLM construction as
a \textquotedblleft black box,\textquotedblright\ but we choose not to do so,
since the properties of the construction that we want are slightly different
from the properties KLM want, and we wish to verify in detail that the desired
properties hold.

\begin{theorem}
[following KLM \cite{klm}]\label{klmthm}Postselected linear optics can
simulate universal quantum computation. \ More concretely: there exists a
polynomial-time classical algorithm that converts a quantum circuit $Q$ over
the gate set $\mathcal{G}$\ into a linear-optics circuit $L$, so that%
\[
\left\langle I\right\vert \varphi\left(  L\right)  \left\vert I\right\rangle
=\frac{\left\langle 0\cdots0\right\vert Q\left\vert 0\cdots0\right\rangle
}{4^{\Gamma}},
\]
where $\Gamma$ is the number of $\operatorname*{CSIGN}$\ gates in $Q$ and
$\left\vert I\right\rangle =\left\vert 0,1,\ldots,0,1\right\rangle $ is the
standard initial state.
\end{theorem}

\begin{proof}
To encode a (qubit-based) quantum circuit by a postselected linear-optics
circuit, KLM use the so-called \textit{dual-rail representation} of a qubit
using two optical modes. \ In this representation, the qubit $\left\vert
0\right\rangle $\ is represented as $\left\vert 0,1\right\rangle $, while the
qubit $\left\vert 1\right\rangle $\ is represented as $\left\vert
1,0\right\rangle $. \ Thus, to simulate a quantum circuit that acts on $k$
qubits, we need $2k$\ optical modes. \ (We will also need additional modes to
handle postselection, but we can ignore those for now.) \ Let the modes
corresponding to qubit $i$ be labeled $\left(  i,0\right)  $\ and\ $\left(
i,1\right)  $\ respectively. \ Notice that the initial state $\left\vert
0\cdots0\right\rangle $\ in the qubit model maps onto the initial state
$\left\vert I\right\rangle $\ in the optical model.

Since $\varphi$\ is a homomorphism by Lemma \ref{homlem}, to prove the theorem
it suffices to show how to simulate the gates in $\mathcal{G}$. \ Simulating a
$1$-qubit gate is easy: simply apply the appropriate $2\times2$\ unitary
transformation to the Hilbert space spanned by $\left\vert 0,1\right\rangle
$\ and $\left\vert 1,0\right\rangle $. \ The interesting part is how to
simulate a $\operatorname*{CSIGN}$ gate. \ To do so, KLM use another gate that
they call $\operatorname*{NS}\nolimits_{1}$, which applies the following
unitary transformation to a single mode:%
\[
\operatorname*{NS}\nolimits_{1}:\alpha_{0}\left\vert 0\right\rangle
+\alpha_{1}\left\vert 1\right\rangle +\alpha_{2}\left\vert 2\right\rangle
\rightarrow\alpha_{0}\left\vert 0\right\rangle +\alpha_{1}\left\vert
1\right\rangle -\alpha_{2}\left\vert 2\right\rangle .
\]
(We do not care how $\operatorname*{NS}\nolimits_{1}$\ acts on $\left\vert
3\right\rangle $, $\left\vert 4\right\rangle $, and so on, since those basis
states will never arise in our simulation.) \ Using $\operatorname*{NS}%
\nolimits_{1}$, it is not hard to simulate $\operatorname*{CSIGN}$ on two
qubits $i$ and $j$. \ The procedure, shown in Figure \ref{ns1fig}, is this:
first apply a Hadamard transformation to modes $\left(  i,0\right)  $\ and
$\left(  j,0\right)  $.%
\begin{figure}[ptb]%
\centering
\includegraphics[
trim=1.105652in 5.144746in 0.742777in 1.107674in,
height=0.7939in,
width=3.6149in
]%
{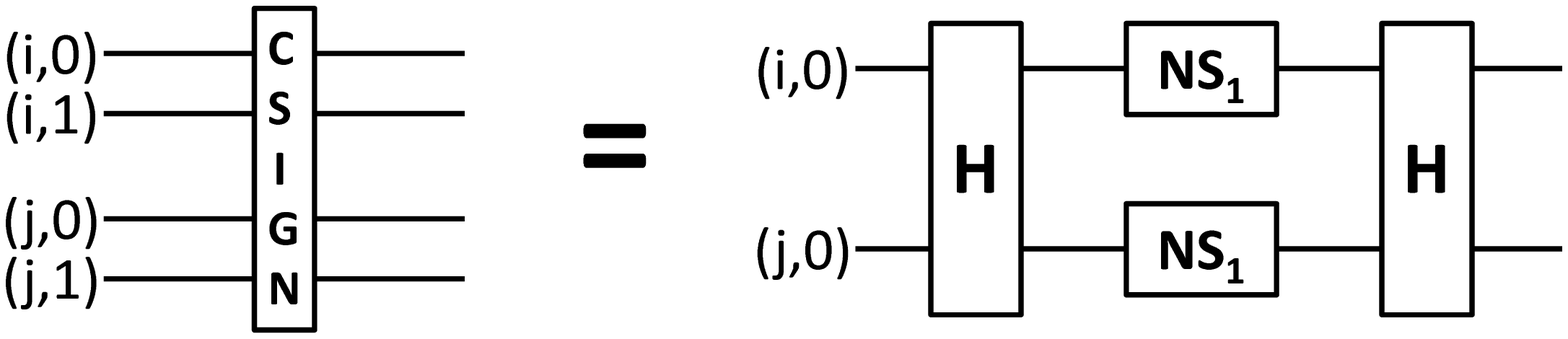}%
\caption{Simulating $\operatorname*{CSIGN}$\ by $\operatorname*{NS}%
\nolimits_{1}$\ and Hadamard.}%
\label{ns1fig}%
\end{figure}
One can check that this induces the following transformation on the state of
$\left(  i,0\right)  $\ and $\left(  j,0\right)  $:%
\begin{align*}
\left\vert 0,0\right\rangle  &  \rightarrow\left\vert 0,0\right\rangle \\
\left\vert 1,0\right\rangle  &  \rightarrow\frac{\left\vert 1,0\right\rangle
+\left\vert 0,1\right\rangle }{\sqrt{2}}\\
\left\vert 0,1\right\rangle  &  \rightarrow\frac{\left\vert 1,0\right\rangle
-\left\vert 0,1\right\rangle }{\sqrt{2}}\\
\left\vert 1,1\right\rangle  &  \rightarrow\frac{\left\vert 2,0\right\rangle
-\left\vert 0,2\right\rangle }{\sqrt{2}}%
\end{align*}
The key point is that we get a state involving $2$ photons in the same mode,
\textit{if and only if} the modes $\left(  i,0\right)  $\ and $\left(
j,0\right)  $\ \textit{both} contained a photon. \ Next, apply
$\operatorname*{NS}\nolimits_{1}$\ gates to both $\left(  i,0\right)  $ and
$\left(  j,0\right)  $. \ This flips the amplitude if and only if we started
with $\left\vert 1,1\right\rangle $. \ Finally, apply a second Hadamard
transformation to $\left(  i,0\right)  $\ and $\left(  j,0\right)  $, to
complete the implementation of $\operatorname*{CSIGN}$.

We now explain how to implement $\operatorname*{NS}\nolimits_{1}$ on a given
mode $i$, using postselection.\ \ To do so, we need two additional modes
$j$\ and $k$, which are initialized to the states $\left\vert 0\right\rangle $
and $\left\vert 1\right\rangle $\ respectively. \ First we apply the following
$3\times3$\ unitary transformation to $i,j,k$:%
\[
W:=\left(
\begin{array}
[c]{ccc}%
1-\sqrt{2} & \sqrt{\frac{3}{\sqrt{2}}-2} & \frac{1}{2^{1/4}}\\
\sqrt{\frac{3}{\sqrt{2}}-2} & \sqrt{2}-\frac{1}{2} & \frac{1}{2}-\frac
{1}{\sqrt{2}}\\
\frac{1}{2^{1/4}} & \frac{1}{2}-\frac{1}{\sqrt{2}} & \frac{1}{2}%
\end{array}
\right)  .
\]
Then we postselect on $j$\ and $k$\ being returned to the state $\left\vert
0,1\right\rangle $. \ As shown in \cite{klm}, this postselection always
succeeds with amplitude $1/2$\ (corresponding to probability $1/4$); and that
conditioned on it succeeding, the effect is to apply $\operatorname*{NS}%
\nolimits_{1}$\ in mode $i$. \ To prove this, observe that since the number of
photons is conserved, the effect of $W$\ on mode $i$ must have the form%
\[
\alpha_{0}\left\vert 0\right\rangle +\alpha_{1}\left\vert 1\right\rangle
+\alpha_{2}\left\vert 2\right\rangle \rightarrow\lambda_{0}\alpha
_{0}\left\vert 0\right\rangle +\lambda_{1}\alpha_{1}\left\vert 1\right\rangle
+\lambda_{2}\alpha_{2}\left\vert 2\right\rangle ,
\]
for some $\lambda_{0},\lambda_{1},\lambda_{2}$. \ Using formula (*), we then
calculate%
\begin{align*}
\lambda_{0}  &  =w_{33}=\frac{1}{2},\\
\lambda_{1}  &  =\operatorname*{Per}\left(
\begin{array}
[c]{cc}%
w_{11} & w_{13}\\
w_{31} & w_{33}%
\end{array}
\right)  =\frac{1}{2},\\
\lambda_{2}  &  =\frac{1}{2}\operatorname*{Per}\left(
\begin{array}
[c]{ccc}%
w_{11} & w_{11} & w_{13}\\
w_{11} & w_{11} & w_{13}\\
w_{31} & w_{31} & w_{33}%
\end{array}
\right)  =-\frac{1}{2}.
\end{align*}
This implies that the $\operatorname*{CSIGN}$\ circuit shown in Figure
\ref{ns1fig}\ succeeds with amplitude $1/4$\ (corresponding to probability
$1/16$), and furthermore, we know when it succeeds.
\end{proof}

In the proof of Theorem \ref{klmthm}, the main reason the matrix $W$\ looks
complicated is simply that it needs to be unitary. \ However, notice that
unitarity is irrelevant for our $\mathsf{\#P}$-hardness application---and if
we drop the unitarity requirement, then we can replace $W$\ by a simpler
$2\times2$\ matrix, such as%
\[
Y:=\left(
\begin{array}
[c]{cc}%
1-\sqrt{2} & \sqrt{2}\\
1 & 1
\end{array}
\right)  .
\]
To implement $\operatorname*{NS}\nolimits_{1}$ on a given mode $i$, we would
apply $Y$ to $i$ as well as another mode $j$ that initially contains one
photon, then postselect on $j$ still containing one photon after $Y$ is
applied. \ One can verify by calculation that the effect on mode $i$ is%
\[
\alpha_{0}\left\vert 0\right\rangle +\alpha_{1}\left\vert 1\right\rangle
+\alpha_{2}\left\vert 2\right\rangle \rightarrow\lambda_{0}\alpha
_{0}\left\vert 0\right\rangle +\lambda_{1}\alpha_{1}\left\vert 1\right\rangle
+\lambda_{2}\alpha_{2}\left\vert 2\right\rangle
\]
where $\lambda_{0}=\lambda_{1}=1$\ and $\lambda_{2}=-1$.

\section{Main Result\label{PROOF}}

In this section we deduce the following theorem, as a straightforward
consequence of Theorem \ref{klmthm}.

\begin{theorem}
\label{mainthm}The problem of computing $\operatorname*{Per}\left(  A\right)
$, given a matrix $A\in\mathbb{Z}^{N\times N}$ of $\operatorname*{poly}\left(
N\right)  $-bit integers written in binary, is $\mathsf{\#P}$-hard\ under
many-one reductions.
\end{theorem}

In classical complexity theory, one is often more interested in various
corollaries of Theorem \ref{mainthm}: for example, that
computing\ $\operatorname*{Per}\left(  A\right)  $\ remains $\mathsf{\#P}%
$-hard\ even if $A$\ is a \textit{nonnegative} integer matrix, or a $\left\{
-1,0,1\right\}  $-valued matrix, or a $\left\{  0,1\right\}  $-valued matrix.
\ Valiant \cite{valiant}\ gave simple reductions by which one can deduce all
of these corollaries from Theorem \ref{mainthm}.\ \ We do not know how to use
the linear-optics perspective to get any additional insight into the corollaries.

Let $C$ be a classical circuit that computes a Boolean function $C:\left\{
0,1\right\}  ^{n}\rightarrow\left\{  -1,1\right\}  $, and let $\Delta
_{C}:=\sum_{x\in\left\{  0,1\right\}  ^{n}}C\left(  x\right)  $. \ Then
computing $\Delta_{C}$, given $C$ as input, is a $\mathsf{\#P}$-hard problem
essentially by definition. \ On the other hand, it is easy to encode
$\Delta_{C}$\ as an amplitude in a quantum circuit:

\begin{lemma}
\label{deltalem}There exists a classical algorithm that takes a circuit $C$ as
input, runs in $\operatorname*{poly}\left(  n,\left\vert C\right\vert \right)
$\ time, and outputs a (qubit-based) quantum circuit $Q$, consisting of gates
from $\mathcal{G}$, such that%
\[
\left\langle 0\cdots0\right\vert Q\left\vert 0\cdots0\right\rangle
=\frac{\Delta_{C}}{2^{n}}.
\]

\end{lemma}

\begin{proof}
Let $D_{C}$\ be\ a $2^{n}\times2^{n}$\ diagonal unitary matrix whose $\left(
x,x\right)  $\ entry is $C\left(  x\right)  $. \ Then since the Toffoli gate
is universal for classical computation, a quantum circuit consisting of
$1$-qubit gates and\ Toffoli gates can easily apply $D_{C}$. \ To do so, one
uses the standard \textquotedblleft uncomputing\textquotedblright\ trick:%
\[
\left\vert x\right\rangle \rightarrow\left\vert x\right\rangle \left\vert
h_{C}\left(  x\right)  \right\rangle \rightarrow C\left(  x\right)  \left\vert
x\right\rangle \left\vert h_{C}\left(  x\right)  \right\rangle \rightarrow
C\left(  x\right)  \left\vert x\right\rangle ,
\]
where $h_{C}\left(  x\right)  $\ is the complete history of a computation
using Toffoli gates that produces $C\left(  x\right)  $. \ Now let
$F=H_{2}^{\otimes n}$\ be the quantum Fourier transform over $\mathbb{Z}%
_{2}^{n}$\ (i.e., the Hadamard gate applied to each of $n$\thinspace\ qubits),
and let $Q=FD_{C}F$. \ Then%
\[
\left\langle 0\right\vert ^{\otimes n}Q\left\vert 0\right\rangle ^{\otimes
n}=\left(  \frac{1}{\sqrt{2^{n}}}\sum_{x\in\left\{  0,1\right\}  ^{n}%
}\left\langle x\right\vert \right)  D_{C}\left(  \frac{1}{\sqrt{2^{n}}}%
\sum_{x\in\left\{  0,1\right\}  ^{n}}\left\vert x\right\rangle \right)
=\frac{\Delta_{C}}{2^{n}}.
\]
Finally, by Lemma \ref{csignlem}, we can simulate each of the Toffoli gates in
$Q$ using gates from the set $\mathcal{G}$.
\end{proof}

Let $Q$ be the quantum circuit from Lemma \ref{deltalem}, and assume $Q$ uses
$k=\operatorname*{poly}\left(  n,\left\vert C\right\vert \right)  $ qubits.
\ By Theorem \ref{klmthm}, we can simulate $Q$ by a linear-optics circuit $L$
such that%
\[
\left\langle I\right\vert \varphi\left(  L\right)  \left\vert I\right\rangle
=\frac{\left\langle 0\cdots0\right\vert Q\left\vert 0\cdots0\right\rangle
}{4^{\Gamma}},
\]
where $\Gamma=\operatorname*{poly}\left(  n,\left\vert C\right\vert \right)  $
is the number of $\operatorname*{CSIGN}$\ gates in $Q$. \ Furthermore, the
circuit $L$ uses $m:=2k+4\Gamma$\ optical modes. \ Let $U$\ be the $m\times
m$\ unitary matrix induced by $L$, and let $V$\ be the $\left(  m/2\right)
\times\left(  m/2\right)  $ submatrix of $U$\ obtained by taking the
even-numbered rows and columns only. \ Then we have%
\begin{align*}
\operatorname*{Per}\left(  V\right)   &  =\left\langle I\right\vert
\varphi\left(  L\right)  \left\vert I\right\rangle \\
&  =\frac{\left\langle 0\cdots0\right\vert Q\left\vert 0\cdots0\right\rangle
}{4^{\Gamma}}\\
&  =\frac{\Delta_{C}}{2^{n}4^{\Gamma}}%
\end{align*}
where the first line follows from formula (*) and the third from Lemma
\ref{deltalem}. \ Since $V$ can be produced in polynomial time given $C$, this
already shows that computing $\operatorname*{Per}\left(  V\right)  $\ to
sufficient precision is $\mathsf{\#P}$-hard.

However, we still need to deal with the issue that the entries of $V$ are real
numbers.\footnote{Indeed, the matrices that we multiply to obtain $U$ can be
\textit{complex} matrices, but $U$ itself (and hence the submatrix $V$) will
always be real.} \ Let $b:=\left\lceil \log_{2}\left(  n!\right)
+2n+2\Gamma\right\rceil $. \ Then notice that truncating the entries of $V$ to
$b$\ bits of precision produces a matrix $\widetilde{V}$\ such that%
\begin{align*}
\left\vert \operatorname*{Per}(\widetilde{V})-\operatorname*{Per}\left(
V\right)  \right\vert  &  \leq n!\left(  1-\left(  1-\frac{1}{2^{b}}\right)
^{n}\right) \\
&  \leq\frac{n!\cdot n}{2^{b}}\\
&  \leq\frac{1}{2^{n+2}4^{\Gamma}}%
\end{align*}
for sufficiently large $n$, and hence%
\[
\left\lfloor 2^{n}4^{\Gamma}\operatorname*{Per}(\widetilde{V})\right\rceil
=2^{n}4^{\Gamma}\operatorname*{Per}\left(  V\right)  =\Delta_{C}.
\]
For this reason, we can assume that each entry of $V$ has the form $k/2^{b}$
for some integer $k\in\left[  -2^{b},2^{b}\right]  $. \ Now set $A:=2^{b}V$.
\ Then $A$ is an integer matrix satisfying $\operatorname*{Per}\left(
A\right)  =2^{bn}\operatorname*{Per}\left(  V\right)  $, whose entries can be
specified using $b+O\left(  1\right)  =\operatorname*{poly}\left(
n,\left\vert C\right\vert \right)  $ bits each. \ This completes the proof of
Theorem \ref{mainthm}.

We conclude by noticing that our proof yields not only Theorem \ref{mainthm},
but also the following corollary:

\begin{corollary}
\label{apxcor}The problem of computing $\operatorname*{sgn}\left(
\operatorname*{Per}\left(  A\right)  \right)  :=\operatorname*{Per}\left(
A\right)  /\left\vert \operatorname*{Per}\left(  A\right)  \right\vert $,
given a matrix $A\in\mathbb{Z}^{N\times N}$ of $\operatorname*{poly}\left(
N\right)  $-bit integers written in binary, is $\mathsf{\#P}$-hard\ under
Turing reductions.
\end{corollary}

\begin{proof}
By the above equivalences, it suffices to show that computing
$\operatorname*{sgn}\left(  \Delta_{C}\right)  $\ is $\mathsf{\#P}$-hard.
\ This is true because, given the ability to compute $\operatorname*{sgn}%
\left(  \Delta_{C}\right)  $, we can determine $\Delta_{C}$ \textit{exactly}
using binary search. \ In more detail, given a positive integer $k$, let
$C\left[  k\right]  $\ denote the circuit $C$ modified to contain $k$
additional inputs $x$ such that $C\left(  x\right)  =1$, and let $C\left[
-k\right]  $ denote $C$ modified to contain $k$ additional $x$'s such that
$C\left(  x\right)  =-1$. \ Then clearly%
\begin{align*}
\Delta_{C\left[  k\right]  }  &  =\Delta_{C}+k,\\
\Delta_{C\left[  -k\right]  }  &  =\Delta_{C}-k.
\end{align*}
Thus we can use the following strategy: compute the signs of $\Delta_{C\left[
1\right]  },\Delta_{C\left[  -1\right]  },\Delta_{C\left[  2\right]  }%
,\Delta_{C\left[  -2\right]  },\Delta_{C\left[  4\right]  },\Delta_{C\left[
-4\right]  },$and so on, increasing $k$ by successive factors of $2$, until a
$k$ is found such that $\operatorname*{sgn}\left(  \Delta_{C\left[  k\right]
}\right)  \neq\operatorname*{sgn}\left(  \Delta_{C\left[  2k\right]  }\right)
$. \ At that point, we know that $\Delta_{C}$\ must be between $k$\ and $2k$.
\ Then by computing $\operatorname*{sgn}\left(  \Delta_{C\left[  3k/2\right]
}\right)  $, we can decide whether $\Delta_{C}$\ is between $k$\ and
$3k/2$\ or between $3k/2$\ and $2k$, and so on recursively until $\Delta_{C}%
$\ has been determined exactly.
\end{proof}

Corollary \ref{apxcor}\ implies, in particular, that approximating
$\operatorname*{Per}\left(  A\right)  $\ to within any multiplicative factor
is $\mathsf{\#P}$-hard---since to output a multiplicative approximation, at
the least we would need to know whether $\operatorname*{Per}\left(  A\right)
$\ is positive or negative.

Using a more involved binary search strategy (which we omit), one can show
that, for any $\beta\left(  N\right)  \in\left[  1,\operatorname*{poly}\left(
N\right)  \right]  $, even approximating $\left\vert \Delta_{C}\right\vert
$\ or $\Delta_{C}^{2}$ to within a multiplicative factor of $\beta\left(
N\right)  $\ would let one compute $\Delta_{C}$\ exactly, and is therefore
$\mathsf{\#P}$-hard under Turing reductions. \ It follows from this that
approximating $\left\vert \operatorname*{Per}\left(  A\right)  \right\vert
$\ or $\operatorname*{Per}\left(  A\right)  ^{2}$\ to within a multiplicative
factor of $\beta\left(  N\right)  $\ is $\mathsf{\#P}$-hard\ as well.
\ (Aaronson and Arkhipov \cite{aark} gave a related but more complicated proof
of the $\mathsf{\#P}$-hardness of approximating\ $\left\vert
\operatorname*{Per}\left(  A\right)  \right\vert $\ and $\operatorname*{Per}%
\left(  A\right)  ^{2}$,\ which did not first replace $\operatorname*{Per}%
\left(  A\right)  $\ with $\Delta_{C}$.)

\section{Acknowledgments}

I am grateful to Alex Arkhipov and Michael Forbes for helpful discussions, and
to Andy Drucker, Greg Kuperberg, Avi Wigderson, Ronald de Wolf, and the
anonymous reviewers for their comments.

\bibliographystyle{plain}
\bibliography{thesis}

\end{document}